\documentclass[12pt]{amsart}
\usepackage{latexsym,amsmath,amsfonts,amscd,amssymb}
\usepackage{graphicx}
\usepackage{xcolor}
\newtheorem{theorem}{Theorem}[section]
\newtheorem{theorem*}{Theorem}
\newtheorem{lemma}[theorem]{Lemma}
\newtheorem{proposition}[theorem]{Proposition}
\newtheorem{proposition*}[theorem*]{Proposition}
\newtheorem{corollary}[theorem]{Corollary}
\newtheorem{corollary*}[theorem*]{Corollary}

\theoremstyle{remark}

\newtheorem{remark*}[theorem*]{Remark}

\newtheorem{note*}[theorem*]{Note}

\newcommand{\EE}{{\mathbb E}}

\newcommand{\NN}{{\mathbb N}}
\newcommand{\PP}{{\mathbb P}}

\newcommand{\RR}{{\mathbb R}}

\newcommand{\ZZ}{{\mathbb Z}}

\def\bitcoinA{%
  \leavevmode
  \vtop{\offinterlineskip 
    \setbox0=\hbox{B}%
    \setbox2=\hbox to\wd0{\hfil\hskip-.03em
    \vrule height .3ex width .15ex\hskip .08em
    \vrule height .3ex width .15ex\hfil}
    \vbox{\copy2\box0}\box2}}

\begin{document}

\title{On Profitability of Nakamoto double spend}

\subjclass[2010]{68M01, 60G40, 91A60m 33B20S.}
\keywords{Bitcoin, blockchain, proof-of-work, selfish mining,  martingale, glambler's ruin, random walk.}

\author{Cyril Grunspan}
\address{L{\'e}onard de Vinci, P{\^o}le Univ., Research Center, Paris-La D{\'e}fense, France}
\email{cyril.grunspan@devinci.fr}
\author{Ricardo P\'erez-Marco}
\address{CNRS, IMJ-PRG,  Paris, France}
\email{ricardo.perez.marco@gmail.com}
\address{\tiny Author's Bitcoin Beer Address (ABBA)\footnote{\tiny Send some anonymous and moderate satoshis to support our research at the pub.}:
1KrqVxqQFyUY9WuWcR5EHGVvhCS841LPLn} 
\address{\includegraphics[scale=0.2]{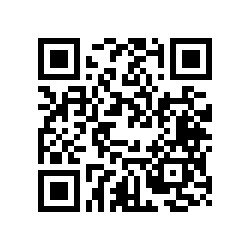}}

\begin{abstract} 
  Nakamoto double spend strategy, described in Bitcoin foundational article, leads to total ruin with positive probability 
  and does not make sense from the profitability point of view. The simplest strategy that can be profitable 
  incorporates a stopping threshold when success is unlikely.
  We solve and compute the exact profitability for this strategy. We compute the minimal amount of the double spend 
  that is  profitable. For a given amount of the transaction, we determine the minimal number of confirmations 
  to be requested by the recipient such that this
  double spend strategy is non-profitable. We find that this number of confirmations is only 1 or 2 
  for average transactions and a small hashrate of the attacker. This is 
  substantially lower  than the original Nakamoto numbers that are widely used and are only based on the success probability instead of the profitability.
\end{abstract}

{\maketitle}

\section{Introduction and background}

Satoshi Nakamoto's Bitcoin foundational article \cite{N08} describes Bitcoin protocol.  
Section 11 contains an analysis of Bitcoin security and 
estimates the probability of success of a double spend attack. 
Unfortunately the strategy proposed has a positive probability of total ruin for the attacker, hence it 
cannot be profitable.

The attacker with a relative hashrate $0<q<1/2$, attemps the double spend by sending 
a legitimate transaction and starts mining a secret fork with a conflicting transaction that invalidates the 
first one. The recipient requests $z\geq 0$ confirmations of the transaction to 
consider it definitive. The goal of the attacker is to catch-up the official blockchain after $z$ confirmations 
of the original transaction. He will only succeed 
with a probability 
that was computed exactly by the authors in \cite{GPM17}, correcting the original approximate 
Nakamoto formula given in \cite{N08}. If he does not succeed, he will be stuck for ever catching-up the official 
blockchain and will go broke. Hence, for a sound strategy, it is necessary to introduce 
a ``give-up'' threshold. If he lags by $A\geq z$ behind the official blockchain the attacker gives-up. Now, this $A$-nakamoto 
strategy (the precise definition in section \ref{sec:A-Nakamoto}) defines an integrable repetition game and 
fits in the general mining profitability 
theory developped by the authors in \cite{GPM2018}. 

According to \cite{GPM2018}, the profitability is compared using  the \textit{Revenue Ratio}
$$
\Gamma=\frac{\EE[\mathbf{R}]}{\EE[\mathbf{T}]}
$$
where $\mathbf{R}$ and $\mathbf{T}$ are random variables, $\mathbf{R}$ is the revenue and $\mathbf{T}$ the duration of the attack. For example,
for the honest strategy consisting of mining one block according to the protocol rules, we have
$\EE[\mathbf{R}_H]=qb$, where $b$ is the coinbase reward, and $\EE[\mathbf{T}_H] =\tau_0$, where 
$\tau_0$ the interblock time\footnote{In the current Bitcoin network $b=12.5 \, \bitcoinA$ and $\tau_0=10 \text{ min}$.}, thus the honest Revenue Ratio is 
\begin{equation*}
\Gamma_H=\frac{qb}{\tau_0} 
\end{equation*}

To compare the profitability of two full time mining strategies is equivalent to compare their Revenue Ratios
(Proposition 3.6 from section 3 in \cite{GPM2018}).
Therefore, the $A$-nakamoto 
strategy is profitable if and only if its Revenue Ratio $\Gamma_A$ is higher than the Revenue Ratio
of the honest strategy, $\Gamma_A>\Gamma_H$.

\medskip

We give an exact closed-form formula for the probability 
of success of the attack.

\medskip

\begin{theorem}\label{propa}
  Let $0<q<1/2$, respectively $p=1-q$, be the relative hashrate of the attacker, 
respectively of honest miners. We denote $\lambda =q/p <1$. 
Let $z\geq 1$ be the number of confirmations requested by the 
recipient of a transaction.  For  $A\geq z$, the probability $P_{A} (z)$ of success for the
$A$-Nakamoto double spend attack is
  $$
    P_{A} (z)  = 
	\frac{I_{4pq}(z,1/2) - \lambda^{A+1}}
	{1 - \lambda^{A+1}} 
  $$
where $I_x(a,b)$ is the Regularized Incomplete Beta function
 $$
 I_x(a,b)=\frac{\Gamma(a+b)}{\Gamma(a) \Gamma(b)} \int_0^x t^{a-1}(1-t)^{b-1} \, dt \ .
 $$
and  $\Gamma$ is Euler Gamma function.
\end{theorem}

\medskip

This generalizes the main result from \cite{GPM17} that we obtain  taking the limit $A\to +\infty$. In the 
formula for $P_A(z)$ we have that  $A+1$ instead of $A$ because we assume that the attacker pre-mines one block
as it is implicit in Satoshi's paper (see Section \ref{sec:A-Nakamoto} below). 

\medskip

\begin{corollary}[G.-P.M., 2017]
The probability of success of the $\infty$-Nakamoto attack is
$$
P_{\infty} (z) = I_{4pq}(z,1/2) \ .
$$
\end{corollary}

\medskip
A major Corollary of this result is obtained taking the asymptotics, 

\begin{corollary}[G.-P.M., 2017]
When $z\to +\infty$, we have
$$
P_\infty(z)\sim \frac{s^z}{\sqrt{\pi (1-s)z}}
$$
where $s=4pq<1$
\end{corollary}

This Corollary is important because it proves the profusely cited and 
``well-known'' result  (but nowhere correctly proved before) 
that this probability decays exponentially to $0$ with the number of confirmations $z$,
hence the probability of a reorganisation of Bitcoin blockchain decays exponential 
with the depth. This is a fundamental result for Bitcoin security.

Observe that $P_A(z)$ decreases with $A$, and $P_A(z)< \lim_{A\to +\infty} P_A(z) = P_\infty (z)$ as expected. 
We also have when $z\to +\infty$
$$
P_A(z)\sim \frac{s^z}{\sqrt{\pi (1-s)z}}
$$
with an asymptotic that is independent of $A\geq z$.
In the next Theorem we make also use of the Beta function 
$$
B(a,b)=\int_0^1 t^{a-1}(1-t)^{b-1} \, dt=\frac{\Gamma(a) \Gamma(b)}{\Gamma(a+b)} \ .
$$
 
The main result in this article is the  computation of the Revenue Ratio 
$\Gamma_A$. 
We compute exact formulas for $\EE[\mathbf{R}_A]$ and $\EE[\mathbf{T}_A]$.


\begin{theorem}\label{thm:main}
With the previous notations, the expected revenue and the expected duration 
of  the $A$-Nakamoto double spend strategy is
    \begin{align*}
	\frac{\EE[\mathbf{R}_A]}{b} &=
	\frac{q z}{2 p} I_{4pq}(z,1/2)
	- \frac{(A+1)\lambda^{A+1}}
	{p(1-\lambda)^3 [A+1]^2} I_{(p-q)^2}(1/2,z) 
	+\frac{2-\lambda+\lambda^{A+2}}{(1-\lambda)^2[A+1]}
	\frac{p^{z-1}q^{z}}{B(z,z)} \\
	& \ \ + P_A(z) \left (\frac{v}{b}+1\right ) \\
   \frac{ \EE[\mathbf{T}_{A}]}{\tau_0} &=
	\frac{z}{2p} I_{4pq}(z,1/2) +
	\frac{A+1}{p(1-\lambda)^2[A+1]}
	I_{(p-q)^2}(1/2,z)  
    - \frac{p^{z-1} q^z}{p(1-\lambda) B(z,z)} 
    +\frac{1}{q}
    \end{align*}  
with the notation $[n]=\frac{1-\lambda^n}{1-\lambda}$ for $n\in\NN$.
\end{theorem}

 As in the original article \cite{GPM2018}, and subsequent application to other block withholding strategies
 \cite{GPM2018}, \cite{GPM2018-2}, \cite{GPM2018-3}, the main tool in the proof of this 
 Theorem are martingale techniques and the application of Doob's Stopping Time Theorem. 
 Previous techniques using Markov chains do not allow 
 the computation of the expected duration of the attack, which is critical for the profitability analysis.
 The profitability analysis is based on attack cycles, modeled by games with repetition, and only holds
 for integrable games, i.e. those that have finite expectation duration of cycles $\EE[\mathbf{T}] <+\infty$
 and allow the application of Doob's Theorem.

 Observe that, all other parameters being fixed, we have the asymptotics when $A\to +\infty$,
 $$
\frac{ \EE[\mathbf{T}_A]}{\tau_0} \sim \frac{I_{(p-q)^2}(1/2,z)}{p-q}\, A
$$
and 
$$
\lim_{A\to +\infty} \frac{\EE[\mathbf{R}_A]}{b} = \frac{\EE[\mathbf{R}_\infty]}{b} = \frac{q z}{2 p}
	I_{4pq}(z,1/2)+\frac{2-\lambda}{1-\lambda}\,\frac{p^{z-1} q^{z}}{B(z,z)}+P_{\infty}(z)\,\left (\frac{v}{b}+1 \right )
$$
In particular we have, 
$$
\lim_{A\to \infty} \EE[\mathbf{T}_A] =\EE [\mathbf{T}_{\infty}] =+\infty
$$
and
$$
\lim_{A\to \infty} \EE[\mathbf{R}_A] = \EE [\mathbf{R}_{\infty}] <+\infty
$$

Hence, in the non-stopping Nakamoto double spend strategy where $A=+\infty$, we have $\Gamma_{\infty} =0$ and any integrable strategy beats Nakamoto non-stopping strategy. Moreover, 
since $\EE [\mathbf{R}_{\infty}] <+\infty$ and $\EE [\mathbf{T}_{\infty}] =+\infty$ Nakamoto's strategy leads to almost sure ruin when considering mining costs.

Another interesting asymptotic is when $q\to 0$, $z\geq 1$,
$$
I_{4pq}(z,1/2)\sim 2\binom{2z-1}{z} q^z
$$
And if we assume $A \geq 2$, $A\geq z\geq 1$, 
$$
\frac{ \EE[\mathbf{R}_A]}{b} \sim \left [2\binom{2z-1}{z} \left (\frac{v}{b}+1\right )
+\frac{2}{B(z,z)} \right ] q^z
$$
and
$$
\frac{ \EE[\mathbf{T}_A]}{\tau_0}  \sim \frac{1}{q} 
$$
Therefore, we have when $q\to 0$,
$$
\Gamma_A \sim \frac{b}{\tau_0} \, \left [2\binom{2z-1}{z} \left (\frac{v}{b}+1\right )
+\frac{2}{B(z,z)} \right ] q^{z+1} 
$$
It is remarkable that this asymptotic is uniform on $A$. 
From this we prove the following practical Corollary. The $A$-Nakamoto double spend is 
profitable when $\Gamma_A \geq \Gamma_H$ and plugging the asymptotics in this profitability 
inequality we get,

\begin{corollary}
When $q\to 0$ the minimal amount to make profitable a Nakamoto double spend with $z\geq 1$ confirmations
is asymptotically
$$
v\geq \frac{q^{-z}}{2\binom{2z-1}{z}} \, b =v_0 \ .
$$
\end{corollary}

For example, with $q=0.01$ and only $z=1$ we need to double spend more than $v_0/b=50$ coinbases. 
For the optimal strategy, the minimal spend for these parameters is $v_0/b = 49.2513$ coinbases as we will see in another article.
With the actual reward of $b=12.5 \, \bitcoinA$ and the actual prize of $\$ 8.600$, this represents more than $\$ 5.375.000$. With  $z=2$ this will be more than $1.666$ coinbases or more 
than more than $179$ million dollar. 
These figures are far from the general belief. Nevertheless, we should note 
that there are sharper strategies and, if ran continuously, we can merge double spend attacks with  
other blockwithholding strategies that exploit the difficulty adjustment formula and this will increase the profitability.

 We observe that since  $\Gamma_A\to 0$ when $A\to +\infty$, there is a value $A_0=A_0(q,v,z)\geq z$ that maximizes the revenue ratio:
$$
\Gamma_{A_0}=\max_{A\geq z} \Gamma_A
$$
Also, $\lim_{z\to +\infty} \Gamma_{A_0} < \Gamma_H$, so
given the amount of the purchase (in coinbase $b$ units), we can compute the number $z$ of confirmations 
that make the $A$-Nakamoto double spend attack non-profitable. This is an important data for the vendor or the 
recipient of the transaction that can set the number of confirmations $z$ accordingly to our formulas.

 We keep the analysis for this simplest strategy as simple as possible. 
 We assume no difficulty adjustment during the attack, and 
 instant block propagation in the network. Other more sophysticated strategies, where $A\leq z$, with 
 important pre-mining, or the optimal strategy, or other hybrid strategies wich combine selfish mining and double spend will be analyzed elsewhere.

 The techniques at the core of the results presented in this article are a combination the techniques developed in \cite{GPM17} and \cite{GPM2018}.

\newpage

\section{Nakamoto double spend strategy.}\label{sec:A-Nakamoto}

\medskip

Let $z$ be the number of confirmations required by the merchant and $v$ the
value of a double spend. We fix a maximal lag $A\geq z$ relative to the public blockchain after which the attacker gives-up.
The relative hashrate of the attacker (resp. honest miners) is $q$ (resp. $p$).
Apparently Nakamoto in \cite{N08} wants to prevent pre-mining by the attacker 
but the generation of fresh keys as he proposes does not prevent it. Also his formulas are only correct
pre-mining one block, for example when he states that the probability 
is $1$ for $z=0$ confirmations. The reader can consult \cite{GPM2019} for a discussion on Section 11 of \cite{N08}. Pre-mining one block is often 
named as a ``Finney attack'' because of the clarification that H. Finney provided in 
2011 (see \cite{Finney} bitcointalk post). We can generalize the preparation of the attack  
by pre-mining an arbitrary number $k$ of blocks before launching the attack.  The precise algorithm employed by the 
attacker in this $(A,k)$-Nakamoto double spend strategy is the following:

\medskip
\medskip

\textit{\centerline{$(A,k)$-Nakamoto double spend strategy}}
\textit{
\begin{itemize}
  \item[0.] Start of the attack cycle (goto $1$).
  \item[1.] The attacker mines honestly on top of the official blockchain $k$ blocks with a transaction that returns the payment funds to an address he controls (goto $2$).
  \item[2.] If  the honest miners take advantage before he pre-mines $k$ blocks, then he restarts mining on top of the 
  new last block of the official blockchain (goto $1$). 
  \item[3.] If he succeeds in pre-mining $k$ blocks lading the honest miners, he keeps his fork secret, sends the purchasing transaction to the vendor, and keeps up mining on his secret fork. (goto $4$).
  \item[4.] If the lag with the official blockchain gets larger than $A$ then the attacker gives-up and the double spend fails (goto $6$).
  \item[5.] If the secret fork of the attacker gets longer than the official blockchain that has added $z$ confirmations to the vendor transaction,  
  then the  attacker releases his fork and the double spend is successful (goto $6$).
  \item[6.] End of the attack cycle (goto $0$).  
\end{itemize}
}

We assume that at $z$ confirmations the attacker receives the goods from the vendor. 
Hence, when the attack is successful, the revenue is $v$ plus 
all the block rewards. When the attack fails, the revenue is $0$ (we assume that he can recover the original bitcoins from the purchase). 
A fundamental observation for the application of the profitability model is that the cost per unit of time is the same as the cost of mining honestly. 
Each time the attacker goes to step $0$, he can start a new attack cycle that ends when he reaches step $6$.

We observe that the strategy has three phases:
\begin{itemize}
  \item The first phase is a pre-mine phase (steps $1-2$).
  \item The attacker sends a conflicting transaction to a merchant and mines on his secret fork until the honest miners have validated $z$ blocks. 
  \item The attacker keeps on mining on his secret fork until his lag is $A$ or his fork gets longer than the official blockchain.
\end{itemize}  

During the second phase, the attacker's lag is always less than $A$ since we consider in this article that we choose $A\geq z$. 
So, the attack cycle cannot terminate before the end of the second phase.
Notice also that there are more general Nakamoto strategies by changing the algorithm in the pre-mining phase, and the last phase. 
The simple $(A,1)$-Nakamoto strategy seems to be the closest profitable strategy 
to the one described by Nakamoto's Bitcoin paper and is the one we study in this article. 
Furthermore, as we will show, there are simple closed forms formulas for this strategy.

The study of the general $(A,k)$-Nakamoto strategy is postponed to another article 
as well as the general optimal strategy attack.

\section{Probability of Success}

We use the same notations and mining model from \cite{GPM17}.
The number of blocks mined by the attacker is a Poisson process $(N'(t))_{t\geq 0}$ and $\mathbf{S}_z$ is the random variable of the time taken to the honest miners to mine $z$ blocks. 
The random variable $N'(\mathbf{S}_z)$  is a $(z,p)$ negative binomial random variable \cite{GPM17}, for $j\geq 0$, 
$$
\PP[N'(\mathbf{S}_z)=j]=p^z q^j \binom{z+j-1}{j}
$$
Recall the basic Euler identity for the Beta function which justifies the Beta distribution,
$$
B(a,b)=\int_0^1 t^{a-1}(1-t)^{b-1} \, dt = \frac{\Gamma(a) \Gamma(b)}{\Gamma(a+b)}
$$
We need a couple more of combinatorial identities stated in the next Lemma.
\begin{lemma}\label{lemnegbin}
For integers $m\geq 1$ and $z\geq 0$, and for  $p, q >0 $ with $q=1-p$, we have,
  \begin{align}
	\label{negbin1}    
    \sum_{j = 0}^{m - 1} p^z q^j \binom{z+j-1}{j} &=  I_p(z,m) \\
    \label{negbin2}
    \sum_{j = 0}^{m - 1} p^z q^j \binom{z+j-1}{j}\cdot j & = 
    \frac{q z}{p} I_p(z,m) -\frac{p^{z-1} q^{m}}{B(z,m)}
  \end{align}
\end{lemma}

\begin{proof}
The first identity is classical (see \cite{AS}, (6.6.3) and (26.5.26), or \cite{DLMF}, (8.17.24), or \cite{GPM17} section 6). 
The second equation follows from the first one differentiating  with respect to $p$,
  $$
\frac{\partial I_p(z,m)}{\partial p}= \frac{z}{p} \, I_p(z,m)-\frac{1}{q} 
\sum_{j = 0}^{m - 1} p^z q^j \binom{z+j-1}{j}\cdot j     
  $$
and observing that  
$\displaystyle\frac{\partial I_p(z,m)}{\partial p}=\frac{p^{z-1} q^{m-1}}{B(z,m)}$\ .
\end{proof}

\begin{proposition}\label{prop:formulas}
If ${\mathbf{X}}$ is a negative binomial random variable with parameters $(p,z)$, then 
we compute, 
  \begin{align}
  \label{negbin3} \sum_{j = 0}^{z - 1} \PP[{\mathbf{X}}=j] & = I_p(z,z)\\
  \label{negbin4} \sum_{j = 0}^{z - 1} \PP[{\mathbf{X}}=j] \left(q/p\right)^{z-j} & = I_q(z,z)\\
  \label{negbin5} \sum_{j = 0}^{z - 1} \PP[{\mathbf{X}}=j]\cdot j & = \frac{q z}{p} I_p(z,z) -\frac{p^{z-1} q^{z}}{B(z,z)}\\
  \label{negbin6} \sum_{j = 0}^{z - 1} \PP[{\mathbf{X}}=j]\,  j \left(q/p\right)^{z-j}& =\frac{p z}{q} I_q(z,z) -\frac{q^{z-1} p^{z}}{B(z,z)}    
  \end{align}
\end{proposition}

\begin{proof}
Identities (\ref{negbin3}) and (\ref{negbin5}) follow from 
Lemma \ref{lemnegbin}. The two other ones follow from these two using, for $j\geq 0$, 
$$
p^z q^j \binom{z+j-1}{j} (q/p)^{z-j} = q^z p^j \binom{z+j-1}{j} 
$$
which means that  $\PP[{\mathbf{X}}=j] \, (q/p)^{z-j}=\PP[{\mathbf{Y}}=j]$
for  ${\mathbf{Y}}$ a $(q,z)$-negative binomial random variable.
\end{proof}
Note also that (\ref{negbin1}) and (\ref{negbin2}) can be restated as 
  \begin{align*}
    \PP[{\mathbf{X}}<m] & = I_p(z,m) \\
    \EE[{\mathbf{X}}|{\mathbf{X}}<m]  & =     
    \frac{q z}{p} -\frac{p^{z-1} q^{m}}{B_p(z,m)}
  \end{align*}
where $B_x(a,b)$ is the incomplete Beta function. Now we prove Theorem \ref{propa}.

\begin{proof}[Proof of Theorem \ref{propa}]
Recall that the attacker has premined one block. So, if he has added $z$ blocks more to his secret fork during the second phase of the attack, then at the end of this phase his secret fork is longer than the official blockchain. In this case, he publishes his fork and
the attack cycle takes end successfully.
Otherwise, the attacker has mined $j$ blocks during the second phase with $j<z$ and he starts a third phase with a lag
of $z-j-1$. The evolution of this lag is a biased random walk 
$({\mathbf Z}_n)$ with a probability $p$ (resp. $q$) to move to the right (resp. left).
The cycle ends when there is $n\in\NN$ such that 
${\mathbf Z}_n = A$ (the attack cycle fails)
or ${\mathbf Z}_n = -1$ (the attack cycle is successful). Hence, according to the gambler's ruin problem formula (see \cite{F}), and using 
formulas (\ref{negbin3}) and (\ref{negbin4}) from Corollary \ref{prop:formulas}, we have
  \begin{align*}
    P_{A} (z) & = \mathbb{P} [N'(\mathbf{S}_z) \geqslant z ] + \sum_{j =
    0}^{z - 1} \mathbb{P} [N'(\mathbf{S}_z) = j] \, \frac{\lambda^{z-j} - \lambda^{A+1}}{1 - \lambda^{A+1}}\\
    & = 1 - \sum_{j = 0}^{z  - 1} \mathbb{P} [N'(\mathbf{S}_z) = j] + \sum_{j =
    0}^{z  - 1} \mathbb{P} [N'(\mathbf{S}_z) = j] \, \frac{\lambda^{z-j} - \lambda^{A+1}}{1 - \lambda^{A+1}}\\
    & = 1 - \left( 1 + \frac{\lambda^{A+1}}{1 - \lambda^{A+1}} \right) \sum_{j = 0}^{z  - 1} \mathbb{P} [N'(\mathbf{S}_z) = j] + \frac{1}{1 - \lambda^{A+1}}  \sum_{j = 0}^{z  - 1} \mathbb{P} [N'(\mathbf{S}_z) = j]\, \lambda^{z - j}\\
    & = 1 - \frac{I_p (z , z)}{1 - \lambda^{A+1}} +
    \frac{I_q(z , z)}{1 - \lambda^{A+1}} 
  \end{align*}
	Finally we use the two classical relations for the incomplete 
	regularized beta function:
	\begin{equation}\label{ipiq}
	I_x(a,b)+I_{1-x}(b,a) = 1
	\end{equation}
for $x\in]0,1[$, $a,b\in\RR_+^*$ and
	\begin{equation}\label{iqi4pq}
	I_q(z,z)= \frac{1}{2} I_{4pq}(z,1/2)
	\end{equation}
See for instance \cite{DLMF} (8.17.4) and (8.17.6).
\end{proof}

\begin{figure}[h]
  \resizebox{300pt}{130pt}{\includegraphics{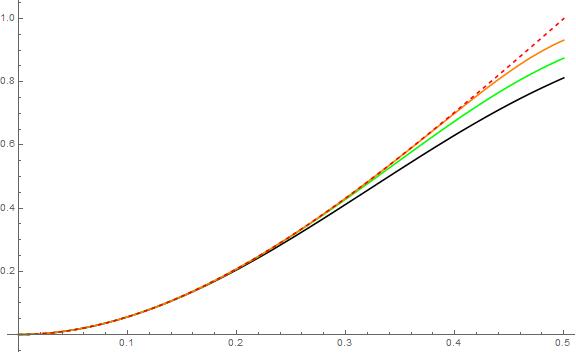}}
  \caption{Graph of $q\mapsto P_A(2)$, $A = 3, 5, 10$ and asymptotics $A\to +\infty$}
\end{figure}

\section{Profitability of the attack.}

\subsection{Expected cycle duration time.}

Note that by definition of the strategy in Section \ref{sec:A-Nakamoto}, 
an attack cycle cannot terminate before the attacker has mined one block (the premined block). So, the duration time of an attack cycle 
$\mathbf{T}$ satisfies $\mathbf{T}=\mathbf{S}'_1+\mathbf{T}'$ where $\mathbf{S}'_1$ is the time before the attacker discovers a new block and $\mathbf{T}'$
is the remaining time of the attack. 
\begin{proposition}
We assume that the attacker has already pre-mined one block. Then, the mean duration time for the end of an attack cycle is 
  $$
    \frac{\EE[\mathbf{T}']}{\tau_0}  = \frac{A+1}{p-q}\cdot 
    \frac{1}{1-\lambda^{A+1}}- \frac{p^{z-1} q^z}{(p-q) B(z,z)} 
    +\left(\frac{z}{p}-\frac{2 (A+1)}{(p-q)(1-\lambda^{A+1})}
    \right) I_q(z,z)
  $$
\end{proposition}

\begin{proof}
We follow the proof of Theorem \ref{propa}.
By definition of the strategy, the attack cycle
cannot end before the honest miners have mined $z$ blocks since we consider in this article that $A\geq z$.
So, $\mathbf{T}' \geq \mathbf{S}_z$ (the initial date $t=0$ is the start of the second phase). Moreover, $\mathbf{T}' = \mathbf{S}_z$
if the attacker has mined $z$ blocks or more during the second phase of the attack. Otherwise, 
the attacker tries to build a fork whose length is greater than the official blockchain, starting with an initial lag of $z-\mathbf{N}'(\mathbf{S}_z)-1$ and gives up 
if this lag becomes greater or equal than $A$ (third phase of the attack). 
So, we have:

$$
\mathbf{T}' = \mathbf{S}_z +{\mathbf{1}}_{\mathbf{N}'(\mathbf{S}_z)<z}\cdot {\mathbf{\tilde T}}_{A+1-\left(z-\mathbf{N}'(\mathbf{S}_z)\right),z-\mathbf{N}'(\mathbf{S}_z)}
$$
with $\mathbf{\tilde T}_{X,Y}={\text{Inf }} \lbrace
t\in\RR_+ ; \left(\mathbf{\tilde N}(t)=\mathbf{\tilde N}'(t)+X\right)
\vee \left(\mathbf{\tilde N}'(t)=\mathbf{\tilde N}(t)+Y\right)
\rbrace$ for $X,Y\in\RR$,
$\mathbf{\tilde N}(t)=\mathbf{N}(t+\mathbf{S}_z)-\mathbf{N}(\mathbf{S}_z)$ and
$\mathbf{\tilde N}'(t)=\mathbf{N}'(t+\mathbf{S}_z)-\mathbf{N}'(\mathbf{S}_z)$. 
By the Markov property, $\mathbf{\tilde N}$ and $\mathbf{\tilde N}'$ are two Poisson processes with parameters $\frac{p}{\tau_0}$ and $\frac{q}{\tau_0}$ independent of $\mathbf{S}_z$ and
$\mathbf{\tilde T}_{X,Y}$ is also independent of $\mathbf{S}_z$.
Moreover, we have:
$$
\frac{\EE[\mathbf{\tilde T}_{X,Y}]}{\tau_0}=\frac{X+Y}{p-q}\left(
\frac{1-\lambda^Y}{1-\lambda^{X+Y}}-\frac{Y}{X+Y}
\right)
$$
This computation is classical and  can be found 
in Appendix A of \cite{GPM2018-3} for example (see Theorem A.1). 
So, we have using Proposition \ref{prop:formulas}
together with (\ref{ipiq}) and (\ref{iqi4pq}),
\begin{align*}
    \frac{\EE[\mathbf{T'}]}{\tau_0} & = \EE[\mathbf{S}_z] + 
    \sum_{j = 0}^{z - 1} 
    \mathbb{P} [\mathbf{N}' (\mathbf{S}_z) = j] \cdot
	\EE[\mathbf{\tilde T}_{A+1-(z-j),z-j}]    
	\\
    & = \frac{z}{p} + \frac{A+1}{p-q}\,  
    \left( \frac{1}{1-\lambda^{A+1}}
	-\frac{z}{A+1}\right) 
	\sum_{j = 0}^{z - 1} \mathbb{P} [\mathbf{N}' (\mathbf{S}_z) = j]
	\\
	&\ \ -\frac{A+1}{p-q}\,  
    \left( \frac{1}{1-\lambda^{A+1}}
	\right) 
	\sum_{j = 0}^{z - 1} \mathbb{P} [\mathbf{N}' (\mathbf{S}_z) = j] \, \lambda^{z-j} 
    \\
    &\ \ + 
    \frac{1}{p-q}\, 
    \sum_{j = 0}^{z - 1} \mathbb{P} 
    [\mathbf{N}' (\mathbf{S}_z) = j]\, j
    \\
    &=  
   \frac{z}{p} + \frac{A+1}{p-q}\, 
    \left( \frac{1}{1-\lambda^{A+1}}
	-\frac{z}{A+1}\right) 
	I_p(z,z)-\frac{A+1}{p-q}\cdot  
    \frac{1}{1-\lambda^{A+1}}
	I_q(z,z) \\
	&\ \ +\frac{1}{p-q}\,  
	\left( 
	\frac{q z}{p} I_p(z,z) -\frac{p^{z-1} q^{z}}{B(z,z)}
	\right)
	\\
    &=
	\frac{A+1}{p-q}\,  
    \frac{1}{1-\lambda^{A+1}}- \frac{p^{z-1} q^z}{(p-q) B(z,z)} 
    +\left(\frac{z}{p}-\frac{2 (A+1)}{(p-q)(1-\lambda^{A+1})}\,  
    \right) I_q(z,z)
    \end{align*}
\end{proof}

\subsection{Expected revenue by cycle.}

\begin{proposition}
The expected revenue per cycle  is 
    \begin{equation*}
	\frac{\EE[\mathbf{R}_A]}{b}
	=
	\frac{q z}{2 p}
	I_{4pq}(z,1/2)
	- \frac{(A+1)\lambda^{A+1}}
	{p(1-\lambda)^3 [A+1]^2} I_{(p-q)^2}(1/2,z)
	+\frac{2-\lambda+\lambda^{A+2}}{(1-\lambda)^2[A+1]}
	\frac{p^{z-1}q^{z}}{B(z,z)}
	+ P_A(z) (v+1)
    \end{equation*}  
    with $[A+1]=\frac{1-\lambda^{A+1}}{1-\lambda}$
\end{proposition}

\begin{proof}
We will use the following notations. If ${\mathbf Z}$ is a biased simple random walk starting at 
${\mathbf Z}_0=k$ with a probability $p$ (resp. $q$) to go right (resp. left), we denote by $\nu_i^k$ with $i\in\ZZ$
the hitting time of $i$ and $\nu_{i,j}^k = \nu_i^k\wedge\nu_j^k$
with $j\in\ZZ$. We also denote by 
${\mathcal{L}}(n)$ the number of steps to the left between $0$ and $n$, that is,
$$
{\mathcal L}(n)=\displaystyle\sum_{i=1}^{n}
{\mathbf 1}_{{\mathbf Z}_{i}={\mathbf Z}_{i-1}-1}.
$$
After the premining phase, the attacker waits for the honest miners to mine $z$ blocks. Suppose that he has mined $j$ blocks during this second phase. If $j\geq z$, then the attack cycle takes end and the attacker has won the double spend amount $v$ and all the $j+1$ blocks he has mined. Otherwise, there is a third phase. The attack cycle still goes on and doesn't
end before the attacker has built a fork whose length 
is greater than the official blockchain or 
his lag became greater or equal than $A$. 
We denote by $\mathbf{Z}_n$ the lag of the attacker plus one at the time when $n$ blocks have been discovered by the attacker or the honest miners since the start of the third phase. Then, 
$\mathbf{Z}_0=z-j$ and $\left(\mathbf{Z}_n\right)_{n\in\NN}$ is a biased simple random walk as above. The attack cycle ends when there is $n$ such that 
$\mathbf{Z}_n=0$ or $\mathbf{Z}_n=A+1$. Therefore, we have:
\begin{align*}
    \frac{\EE[\mathbf{R}_A]}{b}& = \sum_{j=z}^{\infty}
    \PP[\mathbf{N}'(\mathbf{S}_z)=j] (j+1+v)
    +\sum_{j=0}^{z-1}\PP[\mathbf{N}'(\mathbf{S}_z)=j]
	\cdot \PP[\nu_{0,A+1}^{z-j}=\nu_{0}^{z-j}]\\    
    &\cdot (j+1+v+\EE[{\mathcal L}
    (\nu_{0,A+1}^{z-j})|\nu_{0,A+1}^{z-j} = \nu_0^{z-j}])\\
    &=\EE[[\mathbf{N}'((\mathbf{S}_z)]
    -\sum_{j=0}^{z-1}\PP[\mathbf{N}'(\mathbf{S}_z)=j] j
    + P_A(z) (v+1)
	+
	\sum_{j=0}^{z-1}\PP[\mathbf{N}'(\mathbf{S}_z)=j] j
	\cdot \PP[\nu_{0,A+1}^{z-j}=\nu_{0}^{z-j}]
    \\
    &+	\sum_{j=0}^{z-1}\PP[\mathbf{N}'(\mathbf{S}_z)=j]
	\cdot \PP[\nu_{0,A+1}^{z-j}=\nu_{0}^{z-j}]
	\cdot \EE[{\mathcal L}
    (\nu_{0,A+1}^{z-j})|\nu_{0,A+1}^{z-j} = \nu_0^{z-j}] 
    \end{align*}    
Now we use again the classical relation for the gambler's ruin formula
	$\PP[\nu_{0,M}^{m} = \nu_0^{m}]=\frac{\lambda^m - \lambda^M}
	{1-\lambda^M}$ (see for instance \cite{F})
and
\begin{equation*}
\EE[{\mathcal L} (\nu_{0,M}^{m})|\nu_{0,M}^{m} = \nu_0^{m}])=
\frac{m}{2}+
\frac{m\lambda^{m}-\left(2M-m\right)\lambda^{M}
	+\left(2M-m\right)\lambda^{M+m}-m\lambda^{2M}}
	{2 p (1-\lambda) (\lambda^{m}-\lambda^{M})(1-\lambda^{M})}
\end{equation*}
from \cite{GPM2018-3} (See Corollary 2.5) which is a consequence of
Stern's formula \cite{St75}. So, using Proposition \ref{prop:formulas} we compute, 
\begin{align*}
    \frac{\EE[\mathbf{R}_A]}{b}
    &=\frac{q z}{p} + P_A(z) (v+1)\\
	&-\left( 
	\frac{(2(A+1)-z)\lambda^{A+1} + z \lambda^{2(A+1)}}
	{2 p (1-\lambda) (1-\lambda^{A+1})^2}    
	+\frac{\lambda^{A+1}}{1-\lambda^{A+1}}\cdot\frac{z}{2}\right)
	\sum_{j=0}^{z-1}\PP[\mathbf{N}'(\mathbf{S}_z)=j]
	\\
	&\ \ +
	\left(\frac{z+(2(A+1)-z)\lambda^{A+1}}
	{2 p (1-\lambda)(1-\lambda^{A+1})^2}
	+\frac{1}{1-\lambda^{A+1}}\cdot\frac{z}{2}\right)
	\sum_{j=0}^{z-1}\PP[\mathbf{N}'(\mathbf{S}_z)=j] \lambda^{z-j}
	\\	
	&\ \ -
	\left(
	\frac{1}{1-\lambda^{A+1}}
	+\frac{\lambda^{A+1}-\lambda^{2(A+1)}}
	{2 p (1-\lambda)(1-\lambda^{A+1})^2}
	-\frac{\lambda^{A+1}}{1-\lambda^{A+1}}\cdot\frac{1}{2}\right)
	\sum_{j=0}^{z-1}\PP[\mathbf{N}'(\mathbf{S}_z)=j] j
	\\
	&\ \ +	
	\left(
	\frac{1}{1-\lambda^{A+1}}
	-\frac{1-\lambda^{A+1}}
	{2 p (1-\lambda)(1-\lambda^{A+1})^2}
	-\frac{1}{1-\lambda^{A+1}}\cdot\frac{1}{2}\right)
	\sum_{j=0}^{z-1}\PP[\mathbf{N}'(\mathbf{S}_z)=j] j \lambda^{z-j}
\end{align*}    
and
\begin{align*}
    \frac{\EE[\mathbf{R}_A]}{b}
    &=\frac{q z}{p} + P_A(z) (v+1)
	-	\frac{\lambda^{A+1}\left(A+1-q (1-\lambda^{A+1})z\right)}
	{p (1-\lambda) (1-\lambda^{A+1})^2}    
	I_p(z,z)\\
	&+  
    \frac{(A+1)\lambda^{A+1}+p(1-\lambda^{A+1})z}
	{p (1-\lambda)(1-\lambda^{A+1})^2}
	I_q(z,z)
	\\	
	&-
	\frac{p-q+q\lambda^{A+1}}{p (1-\lambda)(1-\lambda^{A+1})}
	\left(
	\frac{q z}{p} I_p(z,z) -\frac{p^{z-1} q^{z}}{B(z,z)}
	\right)-
	\frac{\lambda}{(1-\lambda) (1-\lambda^{A+1})}
	\left(
	\frac{p z}{q} I_q(z,z) -\frac{q^{z-1} p^{z}}{B(z,z)}
	\right)
    \end{align*}    
We note that 
	\begin{equation*}
	\frac{\lambda^{A+1}\left(A+1-q (1-\lambda^{A+1})z\right)}
	{p (1-\lambda) (1-\lambda^{A+1})^2}
	+
	\frac{p-q+q\lambda^{A+1}}{p (1-\lambda)(1-\lambda^{A+1})}
	\cdot
	\lambda z
	= \lambda z + \frac{(A+1)\lambda^{A+1}}
	{p(1-\lambda)(1-\lambda^{A+1})^2}
	\end{equation*}
	So, using again (\ref{ipiq}) and (\ref{iqi4pq}), we get
    \begin{align*}
	\frac{\EE[\mathbf{R}_A]}{b}
	&=
	\left(
	\frac{1}{2}\lambda z + \frac{(A+1)\lambda^{A+1}}
	{p(1-\lambda)(1-\lambda^{A+1})^2}
	\right) I_{4pq}(z,\frac{1}{2})
	+
	P_A(z) (v+1)\\
	&-
	\frac{(A+1)\lambda^{A+1}}
	{p(1-\lambda)(1-\lambda^{A+1})^2}
	+\frac{2-\lambda+\lambda^{A+2}}{(1-\lambda)(1-\lambda^{A+1})}
	\frac{p^{z-1}q^{z}}{B(z,z)}
    \end{align*}    	
\end{proof}

In Figures 2 and 3 we plot the graphs of $q\mapsto \EE[\mathbf{R}_A]$ and $q\mapsto \Gamma_A$
In figure 3 $\Gamma_H$ is the dashed line. 
We have $\lim_{q\to 0.5} \Gamma_{10}(q)=\frac{139}{286}<\frac{1}{2}$ and $\Gamma_{10}(q)<q$ for any $q$. 
So, $(10,1)$-Nakamoto Double Spend strategy with $z=2$ and $v=b$ is
always less profitable than honest mining. 

\begin{figure}[h]
  \resizebox{300pt}{150pt}
  {\includegraphics{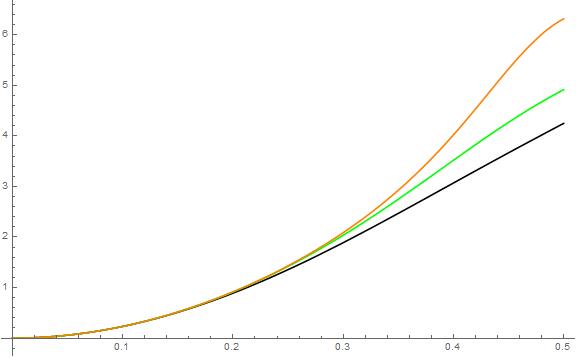}}
  \caption{Graph of $q\mapsto \EE[\mathbf{R}_A]$
  with $z = 2$ and $v=b$ for $A = 3, 5, 10$.}
\end{figure}

\begin{figure}[h]
  \resizebox{300pt}{150pt}
  {\includegraphics{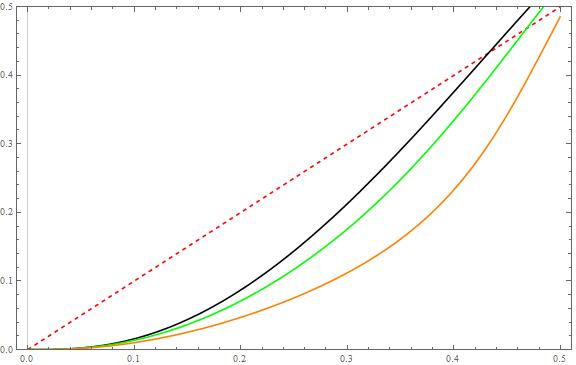}}
  \caption{Graph of $q\mapsto \Gamma_A$, $z = 2$,  $v=b$ for $A = 3,5,10$. }
\end{figure}

\section{Related work}

In \cite{R13} the author proposes the correct formula for the computation of the probability of success of the 
Nakamoto double spend attack from \cite{N08}, which is proved and computed in closed-form using special 
functions in \cite{GPM17}. 
As Corollary of this closed-form formula, it is proved in \cite{GPM17} that the probability decays exponentially to $0$ with 
the number of confirmations $z\to+\infty$. 
In \cite{GZ19} asymptotics at higher orders are computed by more combinatorical methods (classical also from the integral expression in \cite{GPM17}). Also, in this article, the authors have some discussion of 
the initial assumptions of the Nakamoto double spend strategy. It is true that Section 11 of \cite{N08} contains several 
incoherences (we have written a detailed discussion on Section 11 of \cite{N08} in \cite{GPM2019}). All authors consider $z$ to be the number 
of confirmations, which assumes a $1$ block pre-mining (see \cite{R13}, \cite{HT17} or \cite{GPM17}). 
In \cite{SZ16}, the authors look for the best security protocol that a merchant should adopt to counter a double spend attack.
They consider attacks long enough to have an impact on the difficulty adjustment parameter.
 They also propose to merge double spend attacks 
with selfish mining or other blocks witholding strategies 
(it is proven in \cite{GPM2018} that these attacks are profitable on the long term after an adjustment of the difficulty parameter). 
All these articles only consider the double spend attack from the point of view of probabilities.
The duration time of the attack is considered in \cite{GPM2019}. The author computes the conditional probability
density function of the time before an attacker catches up the honest miner knowing that the honest miners have already mined $z$ blocks.
In \cite{HT17}, the authors introduce a profitability setup and
look for the optimal number of blocks that an attacker should pre-mine before launching a double spend attack.
We will answer this question in a forthcoming article.
In \cite{BLOA16}, the authors study the profitability of a 
double-spend attack with a cut-off time strategy,  
$S_{z+1}\wedge S'_{z+1}$ (in our notation). In \cite{JL19} the authors consider a fixed cut-off time  (in case of failure, the attack ends at a fixed time).

What was lacking in the literature is a rigorous model of profitability 
to make exact comparisation of profitabilities of different mining strategies, in particular with the honest strategy, as done in
\cite{GPM2018}. This is what we provide in the present article.

\end{document}